\documentclass[12pt]{article}
\usepackage{amsfonts,epsfig,latexsym,amscd,amsmath,theorem,mathrsfs}

\newcommand{\sdiff}{{\rm SDiff}}

\newcommand{\xvec}{{\bf x}}

\newcommand{\g}{\mathfrak{g}}

\newcommand{\R}{{\mathbb{R}}}
\newcommand{\Z}{{\mathbb{Z}}}

\newcommand{\I}{{\mathbb{I}}}

\newcommand{\beq}{\begin{equation}}
\newcommand{\eeq}{\end{equation}}
\newcommand{\bea}{\begin{eqnarray}}
\newcommand{\eea}{\end{eqnarray}}
\newcommand{\ben}{\begin{eqnarray*}}
\newcommand{\een}{\end{eqnarray*}}
\newcommand{\bem}{\begin{enumerate}}
\newcommand{\eem}{\end{enumerate}}
\newcommand{\ii}{\item}
\newcommand{\ra}{\rightarrow}

\newcommand{\cd}{\partial}
\newcommand{\wt}{\widetilde}

\newcommand{\less}{\backslash}

\newcommand{\hess}{{\sf Hess}}

\newcommand{\su}{{\mathfrak{su}}}

\newcommand{\ph}{{\varphi}}
\def \d{\mathrm{d}}
\newcommand{\dstar}{\delta}
\newcommand{\ip}[1]{\langle #1 \rangle}
\newcommand{\ignore}[1]{}

\newcommand{\RR}{\mathscr{R}}

\renewcommand{\L}{\mathscr{L}}

\newcommand{\vol}{{\rm vol}}

\newcommand{\diag}{{\rm diag}}

\newcommand{\tauvec}{\mbox{\boldmath{$\tau$}}}

\newcommand{\tr}{{\rm tr}\, }

\newcommand{\id}{{\rm id}}

\newcommand{\eps}{\varepsilon}

\renewcommand{\div}{{\rm div}\, }
\renewcommand{\ph}{\varphi}

\theoremstyle{plain}
\newtheorem{thm}{Theorem}
\newtheorem{lemma}[thm]{Lemma}
\newtheorem{prop}[thm]{Proposition}
\newtheorem{cor}[thm]{Corollary}

{\theorembodyfont{\rmfamily}
\newtheorem{defn}[thm]{Definition}
\newtheorem{remark}[thm]{Remark}
\newtheorem{rem}[thm]{Remark}
\newtheorem{eg}[thm]{Example}

}

\newcommand{\news}{\setcounter{equation}{0}}
\newenvironment{proof}{\noindent{\it Proof:\, }}{\hfill$\Box$\vspace*{0.5cm}
}

\renewcommand{\theequation}{\thesection.\arabic{equation}}

\begin{document}

\title{Near BPS Skyrmions and Restricted Harmonic Maps}
\author{
J.M. Speight\thanks{E-mail: {\tt speight@maths.leeds.ac.uk}}\\
School of Mathematics, University of Leeds\\
Leeds LS2 9JT, England}

\maketitle

\begin{abstract}
Motivated by a class of near BPS Skyrme models introduced by Adam, S\'anchez-Guill\'en and Wereszczy\'nski, 
the following variant of the
harmonic map problem is introduced: a map $\ph:(M,g)\ra (N,h)$ between Riemannian manifolds is {\em restricted harmonic} if it locally extremizes 
$E_2$ on its $\sdiff(M)$
orbit, where $\sdiff(M)$ denotes the group of volume preserving diffeomorphisms of $(M,g)$, and $E_2$ denotes the Dirichlet energy. It is conjectured that
near BPS skyrmions tend to restricted harmonic maps in the BPS limit.
It is shown that $\ph$ is restricted harmonic if and only if $\ph^*h$ has exact divergence, and a linear stability theory of restricted harmonic maps
is developed, from which it follows that all weakly conformal maps are stable restricted harmonic. 
Examples of restricted harmonic maps in every degree class
$\R^3\ra SU(2)$ and $\R^2\ra S^2$ are constructed. It is shown that the axially symmetric BPS skyrmions on which all previous analytic studies of near 
BPS Skyrme models have been based, are {\em not}
restricted harmonic, casting doubt on the 
phenomenological predictions of such studies. 
The problem of minimizing $E_2$ for $\ph:\R^k\ra N$
over all {\em linear} volume preserving diffeomorphisms is solved explicitly, and a deformed axially symmetric family of Skyrme fields constructed which
 are candidates for approximate near BPS skyrmions at low baryon number. The notion of restricted
harmonicity is generalized to restricted $F$-criticality where $F$ is any functional on maps $(M,g)\ra (N,h)$ which is, in a precise sense, geometrically
natural. The case where
$F$ is a linear combination of $E_2$ and $E_4$, the usual Skyrme term, is studied in detail, and it is shown that inverse stereographic projection
$\R^3\ra S^3\equiv SU(2)$ is stable restricted $F$-critical for every such $F$.

\end{abstract}

\maketitle

\section{Introduction}

The Skyrme model is an effective theory of nuclear physics in which atomic
nuclei are modelled by topological solitons. It has a single field
$\ph:\R^3\ra SU(2)$ which, by virtue of the boundary condition
$\ph(\infty)=\I_2$, is classified topologically by its degree
$B$, an integer interpreted physically as baryon number. The solitons,
called skyrmions, are the global minimizers, in their degree class, of
an energy functional which, in the standard version of the model, takes the
form
\beq
E(\ph)=\frac12\int_{\R^3}(|\d\ph|^2+|\ph^*\d\mu|^2),
\eeq
where $\mu$ is the left Maurer-Cartan form on the Lie group $SU(2)$. There
is a topological lower energy bound due to Faddeev \cite{fad-bound}
\beq\label{fadbound}
E(\ph)\geq C|B|,
\eeq
but it is known that the bound is never attained \cite{man-skyrme}. Numerical
studies suggest that each degree class has an energy minimizer $\ph_B$, 
and that $E(\ph_B)/B$ is a monotonically decreasing function of $B$. We may
regard $BE(\ph_1)-E(\ph_B)$ as the classical binding energy of a charge
$B$ nucleus, that is, the energy required to break the nucleus into $B$
well-separated individual nucleons. Normalizing this quantity by $E(\ph_1)$,
the rest energy of a single nucleon, one finds that the binding energies
predicted by the standard Skyrme model are much larger than those found
for real nuclei, typically by a factor of around 15 (see table \ref{table1}). This
estimate of nuclear binding energies is, admittedly, rather crude: a more 
refined prediction of nuclear masses requires one
to perform a rather elaborate 
semi-classical quantization of the model. But it seems implausible
that quantum effects will correct such a large discrepancy in the classical
model. 

\begin{table}
\begin{center}
\begin{tabular}{cccc}
$B$&Element&B.E. (Skyrme)&B.E. (experiment)\\ \hline
4&He&0.3639&0.0301\\
7&Li&0.7811&0.0414\\
9&Be&1.0123&0.0615\\
11&B&1.2792&0.0807\\
12&C&1.4277&0.0981\\
14&N&1.6815&0.1114\\
16&O&1.9646&0.1359\\
19&F&2.3684&0.1570\\
20&Ne&2.5045&0.1710\\
\end{tabular}
\end{center}
\label{table1}
\caption{Classical binding energies in the standard Skyrme model, for the first 9 stable composite nuclei. Column 3 shows $(BE_1-E_B)/E_1$, the classical binding 
energy in units of the nucleon mass, computed using data from \cite[p.\ 377]{mansut}. Column 4 shows the same quantity computed from experimental data
\cite{audwap}.}
\end{table}

This has led to considerable recent interest in so-called near-BPS Skyrme
models. The idea is to start with a ``BPS'' Skyrme model, that is, a
Skyrme type model with a linear topological bound like (\ref{fadbound})
which is attained in each degree class. Such a model has exactly zero
(classical) nuclear binding energies. One then perturbs this model in
some way, to obtain a near-BPS model with small but positive
binding energies, in better agreement with nature. One proposal of this type,
due to Sutcliffe \cite{sut-holographic}, starts with pure Yang-Mills
theory on $\R^4$ as the BPS model, reinterprets it as a holographic Skyrme
model on $\R^3$ coupled to an infinite tower of vector mesons, and perturbs
it, to generate a near-BPS model, by truncating the meson tower. This proposal
has many attractive features, not least in providing a satisfying 
explanation for the curious link between skyrmions and instanton holonomies
observed by Atiyah and Manton \cite{atiman}. One disadvantage is that, even
at the lowest truncation level, the model is formidably difficult to simulate
numerically. 

In this paper, we consider a much more direct proposal, due
to Adam, S\'anchez-Guill\'en and Wereszczy\'nski (henceforth ASW)
\cite{adasanwer}. Their idea is to consider an extended Skyrme model
which includes both potential and sextic terms in its energy
\beq
E(\ph)=\frac12\int_{\R^3}\left(c_0U(\ph)^2+c_2|\d\ph|^2
+c_4|\ph^*\d\mu|^2+c_6|\ph^*\vol_{SU(2)}|^2\right)
=:c_0E_0+c_2E_2+c_4E_4+c_6E_6
\eeq
where $c_0,\ldots,c_6\geq 0$ are constants, $U:SU(2)\ra [0,\infty)$
is (the square root of) some potential, and $\vol_{SU(2)}$ denotes the
volume form on $SU(2)$. The standard Skyrme model has $c_0=c_6=0$, and the key
observation here is that the complementary model, with $c_2=c_4=0$, is BPS
\cite{artroctchyan}.
It is convenient to choose length and energy units so that $c_0=c_6=1$. Then,
fixing $c_2=c_4=0$, one finds that
\beq\label{bogbound}
E_{BPS}(\ph)=E_0(\ph)+E_6(\ph)\geq\frac{\ip{U}}{2\pi^2}B
\eeq
with equality if and only if 
\beq
\ph^*\vol_{SU(2)}=U\circ\ph,
\eeq
and fields satisfying this equation (perhaps with rather low regularity)
can be constructed in each degree class. (The constant $\ip{U}$ in equation
(\ref{bogbound}) is the average value of the function $U:SU(2)\ra [0,\infty)$
\cite{spe-semicompactons}.) Note that, given any volume preserving 
diffeomorphism $\psi:\R^3\ra\R^3$, $E_{BPS}(\ph\circ\psi)\equiv
E_{BPS}(\ph)$. Hence, the BPS model also has the
attractive feature of being invariant under the natural action of $\sdiff(\R^3)$, the group of volume preserving diffeomorphisms 
of physical space $\R^3$, so that its skyrmions
can be plastically deformed without changing their energy. This is
reminiscent of the liquid drop model of nuclei.

 Clearly the BPS model by itself is completely unphysical: its
Lorentz invariant extension to Minkowski space has a pathological field equation which does not uniquely define a time evolution of the
field even locally. To get something reasonable, one must at least take $c_2>0$ (whether $c_4>0$ also is a matter of taste). If we assume that
$c_2,c_4$ remain small, however, we should obtain a physically sensible {\em near} BPS model, whose Skyrmions have small binding energy, and
are relatively insensitive to plastic deformations. 

This proposal has been analyzed in detail in a sequence of papers by ASW and collaborators \cite{adasanwer,adasanwer2,adanaysanwer}. 
In \cite{adasanwer2}, the BPS model with the usual pion mass potential 
 is treated, a sequence of axially symmetric BPS skyrmions constructed and a rigid body semi-classical
quantization of these performed to obtain phenomenological predictions of nuclear masses and radii. 
In \cite{adanaysanwer} the rigid body quantization is improved for $B=1$ by including
a dynamical dilation mode, allowing the prediction of so-called Roper resonances. These papers leave $c_2=0$, which is surely
unphysical, and it is not straightforward to translate their results to the case $c_2>0$, small, because, with 
this choice of potential, the BPS skyrmions have infinite $E_2$.  Furthermore, any choice of potential which gives the pions nonzero mass
has the problematic property that the pion mass scales
like $1/\sqrt c_2$, so that, in the near BPS regime, pions are heavier than nucleons. 

These problems were addressed in a pair 
of papers by Marleau and collaborators \cite{bonmar,bonharmar}. In \cite{bonmar} the model with a certain massless potential
is studied. A sequence of axially symmetric charge $B$ BPS skyrmions is constructed which have finite $E_2$ and $E_4$.
 These 
BPS skyrmions are 
 rigid-body quantized within the near BPS model, with $c_2,c_4$ small but nonzero, binding energy curves extracted and $c_2,c_4$
fitted against experimental data. Remarkably, a best fit with $c_4<0$ is proposed, although the model with $c_4<0$ 
has energy unbounded below, and so is unphysical. 
(To see this, note that any degree $0$ field taking values in a two-dimensional submanifold of $SU(2)$
has $E_6\equiv 0$, but $E_4>0$, so generates a family of fields with energy unbounded below under Derrick scaling \cite{der}. Hence $E$ is unbounded below 
in the
$B=0$ sector. Unboundedness in every other sector follows from an obvious gluing construction.
This point  was
also missed in a recent paper by Gudnason and Nitta which, likewise, studies skyrmions with $c_4<0$ \cite{gudnit}.)
 In \cite{bonharmar} a similar analysis is performed, the potential having been tweaked to produce BPS skyrmions with non-shell-like baryon density.
 Once again, best fits with $c_4<0$ are proposed, which is, perhaps, best 
interpreted as suggesting that near BPS Skyrme models with $c_4=0$ are phenomenologically favoured. 

All these 
papers rest on the assumption that in the (physically reasonable) near BPS model, with $c_2>0$ but small, the degree $B$ energy minimizer is well approximated
by the particular axially symmetric BPS skyrmion 
\beq
\ph_B(r,\theta,\phi)=\cos f(B^{-1/3}r)\I_2+i\sin f(B^{-1/3}r)(\sin\theta\cos B\phi,\sin\theta\sin B\phi,\cos\theta)\cdot\tauvec
\eeq
where $f$ is a profile function determined by the potential $U$ and $(r,\theta,\phi)$ are the usual spherical polar coordinates on $\R^3$.
Certainly, it is reasonable to assume that the near BPS skyrmion will be close to {\em some} minimizer of $E_{BPS}$, but
why should it be $\ph_B$? Recall that BPS skyrmions come in infinite dimensional families since, if $\ph$ minimizes $E_{BPS}$, so does every
field in its $\sdiff(\R^3)$ orbit. Consider, for the moment, the case where $c_4$ remains zero. 
Which minimizer $\ph$ of $E_{BPS}$ should we choose to approximate the minimizer of $E_{BPS}+c_2E_2$, where $c_2>0$ is small? Clearly
 $\ph$ should {\em minimize $E_2$ within the space of all degree $B$ minimizers of $E_{BPS}$}. In particular,
$\ph$ should minimize $E_2$ within all fields in its $\sdiff(\R^3)$ orbit. It is not hard to show that, for $|B|>1$, the axially symmetric BPS 
skyrmions $\ph_B$ used in \cite{adasanwer,adasanwer2,adanaysanwer,bonharmar,bonmar} do not have this property, and that their failure to minimize $E_2$
gets worse as $B$ grows.
Recall that a function $\ph:(M,g)\ra (N,h)$ between Riemannian manifolds which locally extremizes the Dirichlet energy $E_2$ 
with respect to {\em all} smooth variations is called a
{\em harmonic map}. What we seek is a map $\ph$ from $\R^3$ to $SU(2)=S^3$, given their usual metrics, which locally extremizes (in fact,
minimizes) $E_2$ not with respect to {\em all} smooth variations, but only with respect to variations arising from volume preserving diffeomorphisms of $(M,g)$.
We say that such a map is {\em restricted harmonic}. 

This paper presents a systematic study of the restricted harmonic map problem in the general setting, before specializing to the case of main interest,
$M=\R^3$, $N=SU(2)=S^3$ with their canonical metrics. It is shown that a map $\ph:(M,g)\ra (N,h)$ is restricted harmonic if and only if $\div\ph^*h$ 
(a one-form on $M$) is exact. The second variation formula for $E_2$ at a restricted harmonic map is derived, yielding a symmetric bilinear form (the hessian)
on the space of divergenceless vector fields on $(M,g)$. A restricted harmonic map is stable if this symmetric bilinear form  is non-negative. It follows
immediately from these formulae that every weakly conformal map is restricted harmonic, stable, and, in fact, locally minimizes $E_2$ on its $\sdiff$ orbit.
For example, inverse stereographic projection $\R^3\ra S^3$ is a stable restricted harmonic skyrme field. We observe that it is also
a BPS skyrmion for an appropriate choice of potential. In fact, every hedgehog skyrme field is restricted harmonic (though stability is an open 
question).
By contrast, the axially symmetric BPS skyrmions $\ph_B$ used in 
\cite{adasanwer,adasanwer2,adanaysanwer,bonmar,bonharmar} are {\em never}
restricted harmonic, for $|B|>1$, so there certainly exist fields in their $\sdiff$ orbits with lower $E_2$, which better approximate near-BPS
skyrmions in the model $E_{BPS}+c_2 E_2$. Another natural family of skyrme fields, those within the rational map ansatz,
can also be shown to have no restricted harmonic members with $|B|>1$.

Constructing the actual $E_2$ minimizer in a given map's $\sdiff$ orbit is a highly nontrivial problem (if, indeed, such a minimizer exists), 
to which we can offer only partial
solutions. First, one can extract from the first variation formula for $E_2$ the direction of steepest descent for $E_2$ tangent to the
$\sdiff$ orbit of $\ph$. This is a divergenceless vector field on $(M,g)$ flow along which, at least initially, improves $\ph$ fastest. Second, 
in the case $M=\R^k$, if one
is (much) less ambitious, and seeks to minimize $E_2$ only over the orbit of the
finite dimensional subgroup of $\sdiff(\R^k)$ consisting of {\em linear} volume preserving 
diffeomorphisms, $SL(k,\R)$, the problem has an easy and neat explicit solution. This allows us to construct a better sequence of maps $\ph_B'=\ph_B\circ A_B$
by linearly deforming those used in previous analytic studies
\cite{adasanwer,adasanwer2,adanaysanwer,bonmar,bonharmar}. These are still not restricted harmonic, but they have much lower $E_2$ than $\ph_B$, particularly at
large $|B|$ ($E_2(\ph_B')\sim |B|^{\frac53}$ whereas $E_2(\ph_B)\sim |B|^{\frac73}$). 

Restricted harmonicity is relevant to near BPS skyrme models with energy $E=E_{BPS}+\eps F$, where the perturbation is
$F=E_2$. More generally, if we consider the model with perturbation
$$
F_\alpha = \alpha E_2+(1-\alpha)E_4,\qquad 0\leq\alpha\leq 1,
$$
then a BPS skyrmion is a sensible approximant to a skyrmion in the perturbed model only if it minimizes 
$F_{\alpha}$
among all maps in its $\sdiff$ orbit. Note that this is just (up to scales) the conventional skyrme energy, which, like $E_2$, has a natural generalization
to the case $\ph:(M,g)\ra (N,h)$ for arbitrary domain and target space \cite{man-skyrme}. By analogy with the harmonic case, $\alpha=1$, we can
define restricted $F$-critical maps (those which locally extremize $F$ on their $\sdiff$ orbit), and derive a linear stability criterion for these, for 
any geometrically natural energy functional $F(\ph)$ (for a precise
definition of ``geometrically natural'' see section \ref{sec-rhm}). The analysis of the case $F=E_2$ generalizes immediately:  $\ph$ is restricted
$F$-critical if and only if $\div S_F$ is exact, where $S_F$ is the stress tensor defined by $F$, and one can find a formula for the
hessian about a restricted $F$-critical map. We apply these formulae in the extreme case $F=F_0=E_4$, to show that inverse stereographic projection is
a stable restricted $F_\alpha$-critical map for all $\alpha\in[0,1]$.

The rest of this paper is structured as follows.  In section
\ref{sec-rhm} restricted harmonic maps are defined in a general geometric setting, and the first variation formula obtained. It is shown that
all hedgehog Skyrme fields are restricted harmonic, that no other fields in the rational map ansatz are, and that all weakly conformal maps are
restricted harmonic. It is also shown that all axially symmetric baby Skyrme fields are restricted harmonic. 
In section \ref{sec-second} the second variation formula is obtained, and it is shown that all weakly conformal maps are
{\em stable}
restricted harmonic.
In section \ref{sec-gen}, restricted $F$-critical
maps are studied for an arbitrary geometrically natural functional $F$, and the case $F=E_4$ (just the Skyrme term) is analyzed in detail. 
Finally, in section \ref{sec-finite}, we solve the finite-dimensional analogue of the restricted harmonic map problem for
maps $\R^k\ra N$, that is, the problem of minimizing $E_2$ over the $SL(k,\R)$ orbit of a given map $\ph$. This produces an
improved family of axially symmetric BPS skyrmions which may be of phenomenological interest.

\section{Restricted harmonic maps}
\news
\label{sec:RHM}
\label{sec-rhm}

Given a minimizer of $E_{BPS}$, we wish to determine whether it minimizes $E_2$ over its orbit under the group of volume preserving diffeomorphisms of physical
space. This is a natural, and apparently novel, variational problem, which makes sense for any smooth map between Riemannian manifolds (an analogous problem
for the Maxwell energy of a magnetic field was considered in \cite{frehe}). In this section we present  a systematic study of this problem in the general
geometric context, before specializing to the original motivating case of BPS skyrmions.

Let $(M,g)$ and $(N,h)$  denote oriented
 Riemannian manifolds of arbitrary dimensions $m$ and 
$n$ respectively, and $\{e_1,\ldots,e_m\}$  be a local orthonormal frame of 
vector fields on $M$. 
Given a vector bundle ${\mathsf E}$ over $M$, $\Gamma({\mathsf E})$ will denote the vector space of smooth sections of ${\mathsf E}$, and
$\odot$ will denote
symmetrized tensor product. $\Omega^p(M)$ will denote the set of smooth $p$-forms on $M$, and $\dstar:\Omega^{p}(M)\ra\Omega^{p-1}(M)$ will be the coderivative, 
that is, the formal $L^2$ adjoint of $\d$
(explicitly, $\dstar=(-1)^{mp+m+1}*\d*$).
The metric $g$ defines canonical isomorphisms between all tensor bundles $T^p_qM$ with the same $p+q$. We will denote by $\flat$ the isomorphism
$T^p_0M\ra T^0_pM$ and by $\sharp$ its inverse. To deal economically with the technicalities arising when $M$ is noncompact, we 
define $\sdiff(M,g)$ to be the space of volume preserving diffeomorphisms of $M$ with compact support (where
the support of a diffeomorphism $\psi:M\ra M$ is the closure of the set $\{x\in M\: :\: \psi(x)\neq x\}$). Note that the formal
tangent space to $\sdiff(M,g)$ at $\id_M$ is the space of smooth divergenceless vector fields of compact support, which we will 
denote $\Gamma_0(TM)$. 
The main definition we want to introduce is the following:

\begin{defn} A smooth map $\ph:(M,g)\ra (N,h)$ is {\em restricted harmonic} if $E_2(\ph)$ is finite and, for all smooth curves $\psi_t$ in $\sdiff(M,g)$ through
$\psi_0=\id_M$,
$$
\frac{d\: }{dt}\bigg|_{t=0}E_2(\ph\circ\psi_t)=0.
$$
\end{defn}

\begin{remark} 
It is clear that, if $E_2(\ph)$ is finite, then so is $E_2(\ph\circ\psi)$ for all $\psi\in\sdiff(M,g)$ since $\ph\circ\psi$ coincides with $\ph$ outside a 
compact set.
Since $\psi_0=\id_M$ we can, without loss of generality, consider only variation curves in $\sdiff_0(M,g)$, the identity component of
$\sdiff(M,g)$. We can then paraphrase the definition as follows: let $E_\ph:\sdiff_0(M,g)\ra\R$ be the function $E_\ph(\psi)=E_2(\ph\circ\psi)$. Then
$\ph$ is restricted harmonic if $\id_M$ is a critical point of $E_\ph$.
The conditions that $E_2(\ph)$ is finite, and that the diffeomorphisms have compact support, are redundant when $(M,g)$ is compact.
\end{remark}

\begin{remark}\label{soisanagi} Clearly $\ph$ harmonic implies $\ph$ restricted harmonic (since $\ph$ is then a critical point of $E_2$ with respect to {\em all}
variations of compact support),
but the converse is false. For example, if $(M,g)=S^1=\R/\Z$ with the
usual metric, then $\vol_g=dx$ and $\sdiff_0$ consists only of the
translation maps $\psi(x)=x+a$. But such maps are isometries, so do not
change $E_2(\ph)$, for any map $\ph:S^1\ra (N,h)$, to any target space 
$(N,h)$. Hence every closed parametrized curve $\ph:S^1\ra(N,h)$ is restricted 
harmonic, whereas only closed {\em geodesics} are harmonic.
\end{remark}

Our first task is to compute the first variation formula associated with
this variational problem. The following definition \cite{spe-crystals} turns out to be
useful for this purpose.

\begin{defn}\label{netbskir} A functional $E(\ph,g)$, which maps each pair consisting of a
smooth map $\ph:M\ra N$ and a Riemannian metric $g$ on $M$ to some
real number, is {\em geometrically natural} if, for all smooth
maps $\ph:M\ra N$, all Riemannian metrics $g$ and all diffeomorphisms
$\psi:M\ra M$,
$$
E(\ph\circ\psi,g)=E(\ph,(\psi^{-1})^*g).
$$
\end{defn}

\begin{remark} In local coordinates on $M$, we can think of a diffeomorphism as a ``passive transformation'', that is, a change of local coordinates. 
Being geometrically natural then reduces to the condition that $E(\ph,g)$ is independent of the choice of local coordinates on $M$. It follows
that all the energy functionals of interest in this paper, $E_0$, $E_2$, $E_4$ and $E_6$, are geometrically natural.
\end{remark}

\begin{remark} For a geometrically natural functional $E(\ph,g)$, a  variation of $\ph$ through diffeomorphisms with $g$ fixed can be reinterpreted as
a variation of the metric $g$ through pullback with $\ph$ fixed. Hence we are led to consider the variation of $E_2(\ph,g)$ with respect to $g$, as
well as $\ph$, and this is encapsulated by the functional's {\em stress tensor}.
\end{remark}

\begin{defn} Let $E(\ph,g)$ be a functional on the space of smooth
maps $\ph:(M,g)\ra N$.
 Let $g_t$ be a smooth curve in the space of Riemannian
metrics on $M$ with $g_0=g$. Let $\eps=\cd_t|_{t=0}g_t$. Note that $\eps$,
like $g$, is a symmetric $(0,2)$ tensor on $M$. The {\em stress tensor}
of $(\ph,g)$ with respect to the functional $E$, is the unique
symmetric $(0,2)$ tensor $S_E(\ph,g)$ on $M$ such that
$$
\frac{d\:}{dt}\bigg|_{t=0}E(\ph,g_t)=\frac12\ip{S_E(\ph,g),\eps}_{L^2}
=\frac12\int_M\ip{S_E(\ph,g),\eps}_g\vol_g.
$$
In the above, we are using the natural inner product between $(0,2)$ tensors
defined by the metric $g$, $
\ip{S,\eps}_g=\sum_{i,j=1}^mS(e_i,e_j)\eps(e_i,e_j)$.
In particular, the stress tensor with respect to $E_2$ is \cite{baieel}
$$
S(\ph,g)=\frac12|\d\ph|^2g-\ph^*h.
$$
\end{defn}

The goal of this section is to give sufficient and necessary conditions for a given map to be restricted harmonic, that is, to compute 
the first variation formula
 for $E_\ph:\sdiff_0(M,g)\ra\R$. 
We will make frequent use of some standard facts about Lie derivatives, which we now summarize
(see \cite{galhullaf} for details).

Let $\psi:M\ra M$ be a diffeomorphism. The {\em push forward}
of a vector field $X\in\Gamma(TM)$ by $\psi$ is the vector field
$\psi_*X(x)=\d\psi_x X(\psi^{-1}(x))$.
The {\em generalized pullback} of $X$ by $\psi$ is $\psi^*X=(\psi^{-1})_*X$,
its push forward by the inverse of $\psi$. The {\em generalized pullback}
of a $(0,1)$ tensor is just its usual
pullback, i.e. $(\psi^*\nu)(X)=\nu(\d\psi X)$.
We extend the generalized pullback to arbitrary $(p,q)$ tensors
by demanding that it has the properties of linearity ($\psi^*(\alpha+\beta)=\psi^*\alpha+\psi^*\beta$, $\psi^*(c\alpha)
=c\psi^*\alpha$) and
distributivity over tensor product 
($\psi^*(\alpha\otimes\beta)=(\psi^*\alpha)\otimes(\psi^*\beta)$).

Any vector field (of compact support) $X$ on $M$ defines a flow
$\Psi:\R\times M\ra M$, $\Psi(t,x)=\psi_t(x)$, such that, for fixed
$x$, $\gamma(t)=\psi_t(x)$ is the integral curve of $X$ with $\gamma(0)=x$. Each map $\psi_t:M\ra M$ is a diffeomorphism of $M$.
The {\em Lie derivative} of a $(p,q)$ tensor $\alpha$ with respect to $X$
is the $(p,q)$ tensor
$$
\L_X\alpha=\frac{\cd\: }{\cd t}\bigg|_{t=0}\psi^*_t\alpha.
$$
This defines, for each $X\in\Gamma(TM)$, a linear operator
$\L_X:\Gamma(T^p_qM)\ra\Gamma(T^p_qM)$. It is immediate from its definition that $\L_X$ preserves the subspaces of
totally symmetric $(0,q)$ tensors $\Gamma(\odot^qT^*M)$ (and $(p,0)$ tensors, $\Gamma(\odot^pTM)$).
\begin{prop}\label{lieprop}
 The operator $\L_X$ has (and is uniquely characterized by)
the following properties:
\bem
\ii For $f\in C^\infty(M)$, $\L_Xf=X[f]=df(X)$.
\ii For $Y\in\Gamma(TM)$, $\L_XY=[X,Y]$.
\ii For all tensors $\alpha,\beta$, $\L_X(\alpha\otimes\beta)=
(\L_X\alpha)\otimes\beta+\alpha\otimes(\L_X\beta)$.
\ii For any $(p,q)$ tensor $\alpha$ and any contraction map 
$c:T^p_qM\ra T^{p-k}_{q-k}M$, $\L_X(c(\alpha))=c(\L_X\alpha)$.
\ii For all $X,Y\in\Gamma(TM)$ and all $\alpha\in\Gamma(T^p_qM)$,
$(\L_X\circ\L_Y-\L_Y\circ\L_X)\alpha=\L_{[X,Y]}\alpha$.
\eem
\end{prop}

It is convenient to extend the definition of divergence from vector fields and one-forms to
arbitrary totally symmetric $(0,q)$ tensors.
\begin{defn} Given a symmetric $(0,q)$ tensor $\alpha$ on $(M,g)$ its 
{\em divergence} is the symmetric $(0,q-1)$ tensor which maps any set
of $q-1$ vector fields $X_2,\ldots,X_{q}$ to the function
$$
(\div \alpha)(X_2,\ldots,X_{q})=
\sum_{i=1}^m(\nabla_{e_i}\alpha)(e_i,X_2,\ldots,X_{q}),
$$
where $\nabla$ is the Levi-Civita connexion on $(M,g)$. Note that $\div g=0$ and, for vector fields $\div\flat X=\div X=-\dstar\flat X$.
\end{defn}

\begin{remark}\label{opatig} It follows from the definition that, for any symmetric $(0,q)$ tensor $\alpha$, and any function
$f\in C^\infty(M)$,
$$
\div(f\alpha)=\iota_{\sharp \d f}\alpha+f\div\alpha,
$$
where $\iota$ denotes interior product,
$(\iota_X\alpha)(X_2,\ldots,X_q)=\alpha(X,X_2,\ldots,X_q)$.
In particular, 
$$
\div fg=\d f.
$$
\end{remark}

The first variation formula will rely on the following lemma, whose proof is presented in the appendix.

\begin{lemma}\label{gsos}
 Let $\alpha$ be a symmetric $(0,2)$ tensor on $(M,g)$ and
$X$ be a vector field. Then
$$
\ip{\alpha,\L_Xg}=2(\div(\iota_X\alpha)-(\div\alpha)(X)).
$$
\end{lemma}

\begin{rem}\label{netski} Putting $\alpha=g$ in Lemma \ref{gsos}, we get the useful fact that
$$
\ip{g,\L_Xg}_g=2\div\iota_Xg=2\div\flat X=2\div X.
$$
Hence the Lie derivative of $g$ along any divergenceless vector field is pointwise othogonal to $g$.
\end{rem}

Having completed these preliminaries, we may now state and prove the first variation formula.

\begin{thm}\label{firstvar}
A smooth map $\ph:(M,g)\ra (N,h)$ of finite Dirichlet energy is restricted harmonic if and only if the one-form
$\div\ph^*h$ on $M$ is exact.
\end{thm}

\begin{proof} Let $\chi_t$ be any smooth curve in $\sdiff(M,g)$ through $\id_M$. Then $\cd_t|_{t=0}\chi_t(x)=X(x)$ is some
divergenceless vector field of compact support, and 
\beq
\frac{d\: }{dt}\bigg|_{t=0}E_2(\ph\circ\chi_t)=\frac{d\: }{dt}\bigg|_{t=0}E_2(\ph\circ\psi_t)
\eeq
where $\psi_t$ is the flow of $X$.
Since $E_2$ is geometrically natural, and $\psi_t$ is a diffeomorphism 
(with inverse
$\psi_{-t}$),
\bea
\frac{d\: }{dt}\bigg|_{t=0}E_2(\ph\circ\psi_t,g)&=&
\frac{d\: }{dt}\bigg|_{t=0}E_2(\ph,\psi_{-t}^*g)\nonumber \\
&=&\frac12\ip{S(\ph,g),\cd_t|_{t=0}\psi_{-t}^*g}_{L^2}\nonumber \\
&=&\frac12\ip{S(\ph,g),-\L_Xg}_{L^2}\label{stophere}\\
&=&-\frac12\int_M 2(\div(\iota_X S)-(\div S)(X))\vol_g\qquad\mbox{(by Lemma
\ref{gsos})}\nonumber \\
&=&\int_M(\div S)(X)\vol_g\qquad\mbox{(since $X$ has compact support)}\nonumber \\
&=&\ip{\flat X,\div S}_{L^2}\nonumber \\
&=&\int_M (\div S)\wedge (*\flat X).
\eea
Now $X$ is divergenceless if and only if $\flat X$ is coclosed. Hence, if $\ph$ is restricted harmonic then
$\div S$ is $L^2$ orthogonal to all {\em coexact} one forms of compact
support $\flat X=\dstar\nu$, so $\d(\div S)=0$ on every compact subset of $M$, and hence $\div S$ is closed.
The integration map
\beq
([\alpha],[\beta])\mapsto \int_M\alpha\wedge\beta
\eeq
defines a nondegenerate pairing $H^1(M)\times H^{m-1}_c(M)\ra\R$ between cohomology classes of closed one-forms, for example
$[\div S]$, and cohomology classes with compact support of closed $(m-1)$ forms, for example $[*\flat X]$,
(the so-called Poincar\'e duality \cite{bottu}). As we have shown above, if $\ph$ is
restricted harmonic then $([\div S],[\beta])\mapsto 0$ for all $[\beta]\in H_c^{m-1}(M)$. Hence, by nondegeneracy of the pairing, $[\div S]=0$,
that is, $\div S$ is exact. But $S=\frac12|\d\ph|^2g-\ph^*h$ so, by Remark \ref{opatig},
\beq
\div S=\d(\frac12|\d\ph|^2)-\div\ph^*h
\eeq
which is exact if and only if $\div\ph^*h$ is exact.

Conversely, if $\div\ph^*h$ is exact, then, as shown above
\beq
\frac{d\: }{dt}\bigg|_{t=0}E_2(\ph\circ\psi_t,g)=0
\eeq
for the flow $\psi_t$ of any divergenceless vector field of compact support, so $\ph$ is restricted harmonic.
\end{proof}

\begin{remark} If $\ph:(M,g)\ra (N,h)$ is restricted harmonic then
\beq\label{PDE}
\d(\div\ph^*h)=0,
\eeq
and if $H^1(M)=0$ the converse holds also (provided $E_2(\ph)$ is finite). So, on a manifold with $H^1(M)=0$, the restricted harmonic map
problem reduces to a nonlinear third-order PDE. If $H^1(M)\neq 0$ then, in addition to solving the PDE (\ref{PDE}), the map $\ph$ must satisfy
a collection of $b_1(M)$ integral constraints
\beq
\int_M\div\ph^*h\wedge\beta_a=0
\eeq
where $\{\beta_a\}$ is a set of generators of $H^{m-1}_c(M)\cong H^1(M)$. In this case, one could describe a map which satisfies  (\ref{PDE}) as
{\em locally} restricted harmonic, since it is critical for volume preserving diffeomorphisms which are trivial outside topologically simple
subsets of $M$. It would be interesting to construct examples of locally restricted harmonic maps that are not restricted harmonic (it is possible that
no such maps exist). 
\end{remark}

\begin{eg}[Restricted harmonic functions]\label{RHF}
We have seen (Remark \ref{soisanagi}) that the restricted harmonic map problem is trivial if the domain is one dimensional (all parametrized curves are
restricted harmonic). The opposite extreme, $\ph:(M,g)\ra\R$, is more interesting. A lengthy but straightforward calculation
shows that
$$
\d(\div\ph^*h)=-(\Delta\d\ph)\wedge\d\ph,
$$
where $\Delta=\d\dstar+\dstar\d$ is the Hodge laplacian on one-forms.
So, on a compact manifold with $H^1(M)=0$, a real function is restricted harmonic if and only if $\Delta\d\ph=f\d\ph$ for some $f:M\ra\R$. Note that any
eigenfunction of $\Delta$ satisfies this condition, and that if $\ph$ satisfies the condition so does $F\circ\ph$ for any smooth $F:\R\ra\R$. This should be
compared with harmonic functions which, on a compact domain, are necessarily constant. 
\end{eg}

\begin{eg}[Axially symmetric baby skyrmions]
Let $\ph:\R^2\ra S^2$ be any map of the form
\beq
\ph(r,\theta)=(\sin f(r)\cos B\theta,\sin f(r)\sin B\theta,\cos f(r))
\eeq
for some profile function $f(r)$ with $f(0)=\pi$, $f(\infty)=0$, and some integer $B$. This is an axially symmetric charge $B$ baby Skyrme
field. We have
\beq
\ph^*h=f'(r)^2dr^2+B^2\sin^2f(r)d\theta^2.
\eeq
Now $g=dr^2+r^2d\theta^2$ is divergenceless, so
\beq
\div dr^2=-\div(r^2d\theta^2)=\frac{1}{r}dr,
\eeq
whence
\beq
\div\ph^*h=2f'(r)f''(r)dr+\left[f'(r)^2-\frac{B^2}{r}^2\sin^2 f(r)\right]\frac{dr}{r}
\eeq
which is closed, hence exact. Hence, every axially symmetric baby skyrme field is restricted harmonic. If $f$ is chosen appropriately (i.e.\ $f(r)=0$ for all
$r\geq r_0$) several such charge $B$ structures can be trivially superposed without overlapping, and the resulting composite field is still restricted harmonic.
Recent numerical work suggests that structures of this type emerge in the $\eps\ra 0$ limit for baby Skyrme models with energy $E=E_0+\eps E_2+E_4$ where the
potential is chosen to support compactons \cite{adajaynayspewer}.
\end{eg}

Recall that a map $\ph:(M,g)\ra (N,h)$ is {\em weakly conformal} if $\ph^*h=fg$ for some function $f:M\ra\R$.

\begin{cor}\label{prscje} Let $\ph:(M,g)\ra (N,h)$ have finite Dirichlet energy and be weakly conformal. Then
$\ph$ is restricted harmonic.
\end{cor}

\begin{proof} By assumption, $\ph^*h=fg$ for some non-negative function
$f:M\ra\R$, so $\div\ph^*h=\d f$ by Remark \ref{opatig}, which is exact.
\end{proof}

\begin{eg}[Suspension Skyrme fields]
\label{barebell}
 Choose and fix a map $\RR:S^2\ra S^2$, and a smooth decreasing function $f:[0,\infty)\ra\R$ with $f(0)=k\pi$, $f(\infty)=0$ where $k\in\Z$.
Then the {\em suspension} of $\RR$ by $f$ is the mapping
\beq\label{suspension}
\phi:\R^3\ra S^3,\qquad \phi(rn)=(\cos f(r),\sin f(r)\RR(n))
\eeq
where $r\geq 0$ and $n\in S^2$. This is a Skyrme field of degree $B=k\deg\RR$ with
\beq\label{sopisanaugir}
\phi^*h=f'(r)^2dr^2+\sin^2 f(r)\RR^*g_{S^2}.
\eeq
 Since $H^1(\R^3)=0$ (or $H^1(\R^3\less\{0\})=0$ if
we make no regularity demand at the origin), such a map is restricted harmonic if and only if $\d(\div\phi^*h)=0$. 

Consider the case where  $\RR$ is holomorphic, so $\ph$ is within the {\em rational map ansatz} \cite[p365]{mansut}. Then $\RR$ is weakly
conformal, so $\ph^*g_{S^2}=\lambda g_{S^2}$ for some function $\lambda:S^2\ra[0,\infty)$. 
Now $\div\, dr^2=2r^{-1}dr$ and $\div g_{S^2}=-2r^{-3}dr$, so
\beq
\div\ph^*h=\left(2ff'+2\frac{(f')^2}{r}-\frac{2}{r^3}\sin^2f\lambda\right)dr+\frac{1}{r^2}\sin^2f\d\lambda.
\eeq
Hence
\beq
\d(\div\ph^*h)=\frac{2}{r^2}\left(\frac{2}{r}\sin f+f'\cos f\right)\d\lambda\wedge dr,
\eeq
so if $\ph$ is restricted harmonic then $r^2\sin f(r)$ is constant or $\lambda(n)$ is constant. The first condition
is incompatible with the boundary conditions for $f$, and the second implies that $R$ is an isometry and hence, up to symmetry,
coincides with $R=\id_{S^2}$. So the only restricted harmonic Skyrme fields in the rational map ansatz are hedgehog fields. Conversely, every hedgehog field 
\beq\label{phiH}
\phi_H(rn)=(\cos f(r),\sin f(r)n)
\eeq
is a $B=k$ restricted harmonic map. Furthermore, given any function $U:S^3\ra [0,\infty)$ which is isospin invariant,
that is, of the form $U(\ph_0,\ph_1,\ph_2,\ph_3)=u(\ph_0)$, and has $u(1)=0$, the BPS Skyrme energy with potential $V=\frac12U^2$ has a $B=1$ minimizer of this form,
with profile function $f(r)$ satisfying the ODE
\beq\label{sapyou}
-\frac{1}{r^2}\frac{df}{dr}\sin^2 f(r)=u(\cos f),\qquad f(0)=\pi,\quad f(\infty)=0.
\eeq
In particular, the model with $u(\ph_0)=(1-\ph_0)^3$ supports the suspension of $\id_{S^2}$ by $f_*(r)=2\cot^{-1}r$ as a $B=1$ BPS skyrmion, and this map
is precisely inverse stereographic projection $\R^3\ra S^3$. Since this is conformal and has finite Dirichlet energy ($E_2=6\pi^2$), we could have deduced
that it is restricted harmonic directly from Corollary \ref{prscje}. 

More generally, the model with potential $V=\frac12(1-\ph_0)^{2\alpha}$, where $\alpha>0$ is a constant,
 has a $B=1$ BPS skyrmion within the hedgehog ansatz. This skyrmion has compact
support if $\alpha<\frac32$, and then occupies total volume
\beq
{\rm Vol}_1=4\pi\int_0^\pi\frac{\sin^2 f}{(1-\cos f)^\alpha}df.
\eeq
It is $C^1$ if $\alpha\geq 1$, and has finite $E_2$ if $\frac12<\alpha<\frac92$. The range $\frac12<\alpha<\frac32$ 
is particularly interesting. In terms of the ball-volume coordinate $v=4\pi r^3/3$, the $k=1$ profile function $f_1(v)$ satisfies the ODE
\beq\label{sayogeso}
\frac{df}{dv}=-\frac{(1-\cos f)^\alpha}{4\pi \sin^2 f}
\eeq
with support $[0,{\rm Vol_1}]$. We can glue together any {\em odd} number, $k$, of copies of $f_1$, using the symmetries $f(v)\mapsto -f(-v)$,
$f(v)\mapsto f(v-c)$ and $f(v)\mapsto f(v)+2\pi$ of (\ref{sayogeso}), to obtain a decreasing profile function $f_{k}(v)$ with $f_{k}(0)=k\pi$ and
$f_k(v)=0$ for all $v\geq k{\rm Vol_1}$, satisfying (\ref{sayogeso}). This is a $B=k$ BPS skyrmion consisting of a charge 1 spherical core surrounded
by $(k-1)/2$ concentric spherical charge 2 shells. As shown above, it is restricted harmonic (though with low regularity). However, its Dirichlet energy
grows like $B^{7/3}$ at large $B$, so this type of BPS skyrmion certainly does not minimize $E_2$ among all BPS solutions of odd charge $B$ for $B$ 
sufficiently large, since it has higher $E_2$ than a superposition of $B$ charge 1 solutions.

Consider now the case where $\RR=\RR_B:S^2\ra S^2$, $\RR_B(\theta,\phi)=(\theta,B\phi)$. Suspension maps of this form
\beq\label{sayosuco}
\ph_B(rn)=(\cos f_B(r),\sin f_B(r)\RR_B(n)),
\eeq
with profile 
function $f_B(r)=f(B^{-1/3}r)$ (with $k=1$), occur frequently  
in studies of near BPS skyrmions \cite{adasanwer,adasanwer2,adanaysanwer,bonmar,bonharmar}. Clearly such fields, for $B>1$, have a string of
conical singularities along the $z$-axis. Nonetheless, provided $f$ satisfies (\ref{sapyou}), they minimize $E_{BPS}$ for the potential $u(\ph_0)$ within the
degree $B$ class. In fact $\ph_B=\ph_H\circ\psi_B$ where $\psi_B:\R^3\less\R_z\ra\R^3\less\R_z$ is the volume preserving $B$-fold covering map
\beq
\psi_B(r(\sin\theta\cos\phi,\sin\theta\sin\phi,\cos\theta))=B^{-1/3}r(\sin\theta\cos B\phi,\sin\theta\sin B\phi,\cos\theta).
\eeq
Unfortunately, none of these BPS skyrmions, for $B>1$, are restricted harmonic, as we now demonstrate. Clearly
\beq
\ph_B^*h=f_B'(r)^2 dr^2 +\sin^2 f_B(r)(d\theta^2+B^2\sin^2\theta d\phi^2).
\eeq
A straightforward calculation in the the local frame
$e_1=\cd_r$, $e_2=r^{-1}\cd_\theta$, $e_3=(r\sin\theta)^{-1}\cd_\phi$ yields
\beq
\div dr^2=\frac{2}{r}dr,\quad
\div d\theta^2=-\frac{1}{r^3}dr-\frac{\cot\theta}{r^2}d\theta,\quad
\div(\sin^2\theta d\phi^2)=-\frac{1}{r^3}dr-\frac{\cot\theta}{r^2}d\theta.
\eeq
Hence
\beq
\div \ph_B^*h=p(r)dr+(B^2-1)\frac{\sin^2 f_B(r)}{r^2}\cot\theta d\theta
\eeq
for a suitably defined function $p(r)$, and this one-form is not closed if $B>1$. It follows that there exist fields in the $\sdiff(\R^3)$ orbit of $\ph_B$ which 
better approximate the minimizer of $E_{BPS}+c_2E_2$ for small $c_2>0$ than $\ph_B$. Note that, like the concentric shell skyrmions described above, these
BPS skyrmions also have energy growth $E_2(\ph_B)\sim B^{7/3}$ at large $B$.
\end{eg}

\begin{remark}[Steepest descent] Let $M$ be compact. Then, by the Hodge isomorphism theorem, there is a unique $L^2$ orthogonal
decomposition of the one-form $\div\ph^*h$ as
\beq
\div\ph^*h=\nu_{harmonic}+\nu_{coexact}+\nu_{exact}=\nu_{coclosed}+\nu_{exact}
\eeq
and $\ph$ is restricted harmonic if and only if $\nu_{coclosed}=0$. From the proof of Theorem \ref{firstvar}, we see that the rate of change of 
$E_\ph(\psi)=E_2(\ph\circ\psi)$ along $X\in T_{\id_M}\sdiff(M,g)=\Gamma_0(TM)$ is
\beq
\d E_\ph(X)=\ip{\flat X,\div S}_{L^2}=-\ip{\flat X,\div\ph^*h)}_{L^2}=-\ip{\flat X,\nu_{coclosed}}_{L^2}
\eeq
since $\dstar\flat X=0$. Hence, the direction of steepest descent of $E_\ph$ at $\id_M$ is
\beq
X_{steepest}=\sharp\nu_{coclosed}.
\eeq
To construct this, we seek a function $f:M\ra\R$ such that $\div\ph^*h -\d f$ is coclosed, that is,
\beq
\Delta f=\dstar(\div\ph^*h).
\eeq
Given the solution to this Poisson equation (which is unique up to an additive constant),
\beq
X_{steepest}=\sharp(\div\ph^*h-\d f).
\eeq
In practice, the easiest way to construct a divergenceless vector field is to write down a {\em potential} for it, i.e.\ $\omega\in\Omega^2(M)$
such that $X=\sharp\dstar\omega$. If $H^1(M)=0$ then all divergenceless vector fields arise in this way. The rate of change of $E_\ph$ along the
vector field generated by potential $\omega$ is
\beq
(\d E_\ph\circ\sharp\circ\dstar)(\omega)=-\ip{\omega,\d(\div\ph^*h)}_{L^2}.
\eeq
Hence the potential $\omega$ which gives steepest descent, for fixed $\|\omega\|_{L^2}$, is in the direction
\beq
\omega_{steepest}=\d(\div\ph^*h).
\eeq
Note that $\sharp\dstar\omega_{steepest}\neq X_{steepest}$ in general, since $\sharp\delta:\Omega^2(M)\ra\Gamma_0(TM)$ is not an $L^2$ isometry.
\end{remark}

\section{Second variation formula and stability}
\label{sec-second}

\news

Recall that $\ph$ is restricted harmonic if it is a critical point of
$E_2$ restricted to its $\sdiff_0$ orbit. Given such a critical point, 
it is natural to ask about its {\em stability}, that is, whether it is
a local minimum of energy (stable), or merely a saddle point (unstable).
To answer this, one must compute the {\em second variation} of the energy
about the critical point to obtain, in analogy with standard harmonic map theory, its hessian \cite[p.\ 91]{baiwoo}:

\begin{defn} Let $\ph:(M,g)\ra (N,h)$ be restricted harmonic and
$X,Y$ be any pair of divergenceless vector fields on $M$.
Let $\psi_{s,t}$ be a two-parameter variation of $\psi_{0,0}=\id_M$
in $\sdiff_0(M,g)$ tangent to $X,Y$, that is, with $X=\cd_s\psi_{s,t}|_{s=t=0}$,  
$Y=\cd_t\psi_{s,t}|_{s=t=0}$.
The {\em hessian} of $E_2$ at $\ph$ is the bilinear form
$$
\hess:\Gamma_0(TM)\times \Gamma_0(TM)\ra \R,\qquad
\hess(X,Y)=\frac{\cd^2 E_2(\ph\circ\psi_{s,t})}{\cd s\: \cd t}\bigg|_{s=t=0}.
$$
We say that $\ph$ is {\em stable} if $\hess(X,X)\geq 0$ for all $X$,
 and {\em unstable} otherwise.
\end{defn}

To compute an explicit formula for $\hess$, it is useful to have an alternative formulation
of the first variation.

\begin{lemma}\label{firstvar2} Let $\ph:(M,g)\ra (N,h)$ be restricted harmonic. 
Then $\ip{\ph^*h,\L_Xg}_{L^2}=0$ for all $X\in\Gamma_0(TM)$.
\end{lemma}

\begin{proof}
Let $\psi_t$ be the flow of $X\in\Gamma_0(TM)$. Following the proof of Theorem \ref{firstvar} to line (\ref{stophere}), we have that
\beq
\frac{d\: }{dt}\bigg|_{t=0}E_2(\ph\circ\psi_t)=-\frac12\ip{S,\L_Xg}_{L^2}
=-\frac12\ip{\frac12|\d\ph|^2g-\ph^*h,\L_Xg}_{L^2}=\frac12\ip{\ph^*h,\L_Xg}_{L^2}
\eeq
by Remark \ref{netski}.
\end{proof}

\begin{thm}\label{secondvar}
 Let $\ph$ be restricted harmonic. Then
$$
\hess(X,Y)=\frac12\ip{\L_X\ph^*h,\L_Yg}_{L^2}.
$$
\end{thm}

\begin{proof} Given $X,Y\in\Gamma_0(TM)$, let $\psi_s$, $\chi_t$ denote their 
flows, and choose $\psi_{s,t}=\psi_s\circ\chi_t$ as the two-parameter variation
of $\id_M$ in $\sdiff_0(M,g)$ tangent to them. Let $\ph_s=\ph\circ\psi_s$.
Then
\bea
\hess(X,Y)&=&\frac{\cd^2 E_2(\ph\circ\psi_{s}\circ\chi_t)}{\cd s\: \cd t}\bigg|_{s=t=0}
=\frac{\cd^2 E_2(\ph_s\circ\chi_t,g)}{\cd s\: \cd t}\bigg|_{s=t=0}\nonumber \\
&=&\frac{\cd\:}{\cd s}\bigg|_{s=0}\frac{\cd\:}{\cd t}\bigg|_{t=0} E_2(\ph_s,\chi_{-t}^*g)\qquad
\mbox{(since $E_2$ is geometrically natural)}\nonumber \\
&=&\frac{d\: }{ds}\bigg|_{s=0}\frac12\ip{S(\ph_s,g),-\L_Yg}_{L^2}.
\eea
Now
\bea
S(\ph_s,g)&=&\frac12|\d\ph_s|^2g-\ph_s^*h\nonumber \\
\Rightarrow\quad
\frac{d\: }{ds}\bigg|_{s=0}S(\ph_s,g)&=&fg-\frac{d\: }{ds}\bigg|_{s=0}\psi_s^*(\ph^*h)
=fg-\L_X\ph^*h,
\eea
and, by Remark \ref{netski}, $fg$ is pointwise orthogonal to $\L_Yg$ (since $\div Y=0$).
The result immediately follows.
\end{proof}

\begin{cor}\label{keanwarimyco} Let $\ph$ be weakly conformal with finite $E_2$. Then $\ph$ is a stable restricted harmonic map.
\end{cor}

\begin{proof} We have seen (Corollary \ref{prscje}) that $\ph$ is restricted harmonic.
By assumption, $\ph^*h=f^2g$ for some $f\in C^\infty(M)$, so
\beq
\hess(X,X)=\frac12\ip{f^2\L_Xg+2fX[f]g,\L_Xg}_{L^2}=\|f\L_Xg\|_{L^2}^2
\eeq
by Remark \ref{netski}.
\end{proof}

In particular the hedgehog BPS skyrmion for potential $V=\frac12(1-\ph_0)^6$ (inverse stereographic projection) is a stable restricted harmonic map.
In fact, we see that, since weakly conformal maps have isolated critical points, $\hess(X,X)>0$
unless $X$ is Killing ($\L_Xg=0$), that is, generates an isometry. Hence, weakly conformal maps are
local minima of $E_2$ on the homogeneous space $\sdiff_0\cdot\ph/{\rm Isom}\cdot \ph$.

\begin{rem} It is possible for an {\em unstable} harmonic map to be a {\em stable} restricted harmonic map, if the space of
unstable variations is $L^2$ orthogonal to $\d\ph(\Gamma_0(TM))$. For example, the identity map $\id:S^n\ra S^n$ is an unstable harmonic map for $n\geq 3$
\cite{smi},
but is conformal, so is stable as a restricted harmonic map.
\end{rem}

\begin{rem} By its definition, $\hess$ should be a {\em symmetric} bilinear form on
$\Gamma_0(TM)$. Our formula for it is not manifestly symmetric, so, as a consistency check, we should
verify that $\hess(X,Y)=\hess(Y,X)$ directly from our formula.
\end{rem}

\begin{defn} Given a symmetric $(0,2)$ tensor $\alpha$ on $(M,g)$, denote by $\sharp\alpha$ the
symmetric $(2,0)$ tensor metrically dual to $\alpha$ with respect to $g$. Explicitly,
\beq
\sharp\alpha=\sum_{i,j}\alpha(e_i,e_j)e_i\otimes e_j.
\eeq
Conversely, given a symmetric $(2,0)$ tensor $\beta$, denote by $\flat\beta$ the symmetric
$(0,2)$ tensor metrically dual to $\beta$. Clearly $\flat\sharp\alpha=\alpha$ and $\sharp\flat\beta
=\beta$. Furthermore, $\ip{\wt\alpha,\alpha}=c(\sharp\wt\alpha\otimes\alpha)$ where $c$ denotes the
unique contraction $\odot^2TM\otimes\odot^2T^*M\ra\R$.
\end{defn}

\begin{lemma}\label{LXadj}
 Let $X$ be a divergenceless vector field of compact support on $(M,g)$. Then the formal $L^2$ adjoint
of the Lie derivative operator $\L_X:\Gamma(T^*M\odot T^*M)\ra\Gamma(T^*M\odot T^*M)$ is
$$
\L_X^\dagger=-\flat\L_X\sharp.
$$
\end{lemma}

\begin{proof} We want to show that, for all symmetric $(0,2)$ tensors $\alpha,\wt\alpha$,
$$
\ip{\wt\alpha,\L_X\alpha}_{L^2}=-\ip{\flat\L_X\sharp\wt\alpha,\alpha}_{L^2}.
$$
Now, $\ip{\wt\alpha,\alpha}=c(\sharp\wt\alpha\otimes\alpha)$, where $c$ denotes contraction, so
\bea
\div(\ip{\wt\alpha,\alpha}X)&=&X[\ip{\wt\alpha,\alpha}]+\ip{\wt\alpha,\alpha}\div X\nonumber \\
&=&\L_X(c(\sharp\wt\alpha\otimes\alpha))+0\nonumber \\
&=&c(\L_X(\sharp\wt\alpha\otimes\alpha))\nonumber \\
&=&c((\L_X\sharp\wt\alpha)\otimes\alpha+\sharp\wt\alpha\otimes\L_X\alpha)\nonumber \\
&=&\ip{\flat\L_X\sharp\wt\alpha,\alpha}+\ip{\wt\alpha,\L_X\alpha}.
\eea
where we have repeatedly used Proposition \ref{lieprop}.
Now integrate both sides over $M$ and use the divergence theorem.
\end{proof}

\begin{defn} Given any pair of $(0,2)$ tensors $A,B$ on $(M,g)$, their {\em dot product} is
the $(0,2)$ tensor $A\cdot B$ defined by
$$
(A\cdot B)(X,Y)=\sum_{i=1}^mA(X,e_i)B(e_i,Y).
$$
\end{defn}

\begin{lemma}\label{th}
 If $X$ is a vector field and $\alpha\in\Gamma(T^*M\odot T^*M)$, then
$$
\flat\L_X\sharp\alpha=\L_X\alpha-\alpha\cdot\L_Xg-\L_Xg\cdot\alpha.
$$
\end{lemma}

\begin{proof}
Let $\alpha_{ij}=\alpha(e_i,e_j)$. Then, by definition, 
\bea
\sharp\alpha&=&\sum_{i,j}\alpha_{ij}e_i\otimes e_j\nonumber \\
\Rightarrow\quad
\L_X\sharp\alpha&=&\sum_{i,j}\left(X[\alpha_{ij}]e_i\otimes e_j+\alpha_{ij}\L_Xe_i\otimes e_j
+\alpha_{ij}e_i\otimes\L_Xe_j\right).
\eea
Let $\beta=\flat\L_X\sharp\alpha$. Then
\bea
\beta(e_k,e_l)&=&X[\alpha_{kl}]+\sum_i(\alpha_{il}g(\L_Xe_i,e_k)+\alpha_{ki}g(\L_Xe_i,e_l))\nonumber \\
&=&(\L_X\alpha)(e_k,e_l)+\alpha(\L_Xe_k,e_l)+\alpha(e_k,\L_Xe_l)+\sum_i(\alpha_{il}g(\L_Xe_i,e_k)+\alpha_{ki}g(\L_Xe_i,e_l))\nonumber \\
&=&(\L_X\alpha)(e_k,e_l)+\sum_i(\alpha_{il}(g(\L_Xe_i,e_k)+g(e_i,\L_Xe_k))+\alpha_{ki}(g(\L_Xe_i,e_l)+g(e_i,\L_Xe_l))\nonumber \\
&=&(\L_X\alpha)(e_k,e_l)-\sum_i(\alpha_{il}\L_Xg(e_i,e_k)+\alpha_{ki}\L_Xg(e_i,e_l))\nonumber \\
&=&\left(\L_X\alpha-\L_Xg\cdot \alpha-\alpha\cdot\L_Xg\right)(e_k,e_l)
\eea
as was to be proved.
\end{proof}

\begin{prop} Let $\ph:(M,g)\ra (N,h)$ 
be restricted harmonic and $\hess$
be the bilinear form defined in Theorem \ref{secondvar},
$$
\hess:\Gamma_0(TM)\times\Gamma_0(TM)\ra \R,\qquad \hess(X,Y)=\frac12\ip{\L_X\ph^*h,\L_Yg}_{L^2}.
$$
Then $\hess$ is symmetric.
\end{prop}

\begin{proof} By Lemmas \ref{LXadj} and \ref{th},
\bea
\hess(X,Y)&=&-\frac12\ip{\ph^*h,\flat\L_X\sharp\L_Yg}_{L^2}\nonumber \\
&=&-\frac12\ip{\ph^*h,\L_X\L_Yg-\L_Yg\cdot\L_Xg-\L_Xg\cdot\L_Yg}_{L^2}.
\eea
Hence, by Proposition \ref{lieprop},
\beq
\hess(X,Y)-\hess(Y,X)=-\frac12\ip{\ph^*h,\L_{[X,Y]}g}_{L^2}.
\eeq
Now the space of divergenceless vector fields is closed under Lie bracket, so $\ph^*h$ is $L^2$
orthogonal to $\L_{[X,Y]}g$ by Lemma \ref{firstvar2}.
\end{proof}

\section{More general perturbations}
\news
\label{sec-gen}

So far, we have considered near BPS Skyrme models of the form $E=E_{BPS}+\eps E_2$, that is, where the BPS energy functional $E_{BPS}=E_0+E_6$
is perturbed by adding 
a small constant times the Dirichlet energy. This led us to consider the problem of minimizing $E_2$ over symmetry orbits of minimizers of $E_{BPS}$. 
It is interesting to consider more general perturbations of the form $E=E_{BPS}+\eps F$, where $F=E_2+E_4$, for example. In fact, the local theory
developed in section \ref{sec:RHM} generalizes immediately to this setting, provided $F$ is geometrically natural, in the sense of
Definition \ref{netbskir}.

\begin{defn} A smooth map $\ph:(M,g)\ra (N,h)$ is {\em resticted $F$-critical} if $F(\ph)$ is finite and, for all smooth curves $\psi_t$ in $\sdiff(M,g)$ through
$\psi_0=\id_M$,
$$
\frac{d\: }{dt}\bigg|_{t=0}F(\ph\circ\psi_t)=0.
$$
\end{defn}

\begin{thm}\label{firstvarF} Let $F(\ph,g)$ be a geometrically natural functional with stress tensor $S_F$.
A smooth map $\ph:(M,g)\ra (N,h)$ of finite $F$ is restricted $F$-critical if and only if the one-form
$\div S_F$ on $M$ is exact.
\end{thm}

\begin{proof} Follows {\it mutatis mutandis} the proof of Theorem \ref{firstvar}, with $E_2$ replaced by $F$.
\end{proof}

As for $E_2$, we can analyze the stability of restricted $F$-critical maps by computing the second variation
\beq
\hess_F(X,Y)=\frac{\cd^2\: \: }{\cd s\cd t}\bigg|_{s=t=0}F(\ph\circ\psi_{s,t})
\eeq
for any two-parameter variation $\psi_{s,t}$ of $\id_M$ in $\sdiff_0(M,g)$ tangent to $X,Y\in\Gamma_0(TM)$. The resulting
formula for $\hess_F$ depends on the details of $S_F$. 

To illustrate, consider the case $N=G=SU(2)$ and $F=E_4$, the Skyrme term, 
\beq
E_4=\frac12\|\ph^*\omega\|_{L^2}^2,
\eeq
where $\omega$ is the $\g=\su(2)$-valued two-form on $G$
\beq
\omega(X,Y)=[\mu(X),\mu(Y)]_\g,
\eeq
and $\mu\in\Omega^1(G)\otimes\g$ is the left Maurer-Cartan form. This has
 stress tensor \cite{spe-crystals}
\beq
S_{E_4}(\ph,g)=\frac12|\ph^*\omega|^2g+\ph^*\omega\cdot\ph^*\omega,
\eeq
where the final term denotes the (real valued) symmetric $(0,2)$ tensor 
\beq
\ph^*\omega\cdot\ph^*\omega(X,Y)=\sum_i \ip{\ph^*\omega(X,e_i),\ph^*\omega(e_i,Y)}_\g.
\eeq
(For an alternative characterization of $E_4$ and its stress tensor, which avoids using the lie group structure of the target space, see
\cite{slo}.)
A Skyrme field $\ph:M\ra N$ is restricted $E_4$-critical if and only if $\div\, S_F$ is exact, and hence, if and only if $\div(\ph^*\omega\cdot\ph^*\omega)$
is exact.

For example, let $\ph:\R^3\ra N$ be a hedgehog field
\beq
\ph(r,n)=\cos f(r)\I_2+i\sin f(r) n\cdot \tau.
\eeq
We can identify $\g$ with $\R^3$ given the Lie bracket $[u,v]=-2u\times v$, where $\times$ denotes vector product. Then
\bea
\ph^*\mu(\cd/\cd r)&=&f'(r)n\nonumber\\
\ph^*\mu(u)&=&\sin f(\cos f u+\sin f n\times u)\qquad\mbox{for all $u\in T_nS^2$},
\eea
and hence
\bea
\ph^*\omega(\cd/\cd r, u)&=&-2f'\sin f(-\sin f u+\sin f n\times u)\nonumber \\
\label{kelannwat}
\ph^*\omega(u,n\times u)&=&-2\sin^2f|u|^2 n.
\eea
It follows that
\bea
\ph^*\omega\cdot\ph^*\omega&=&-\frac{8}{r^2}(f'\sin f)^2dr^2-4\sin^2f\left\{(f')^2+\frac{1}{r^2}\sin^2f\right\}g_{S^2}\nonumber\\
&=&p(r)dr^2+q(r)g
\eea
for some functions $p(r),q(r)$ of $r$ only. This has closed, hence exact, divergence.
Hence, every hedgehog field is both restricted harmonic and restricted $E_4$-critical, so restricted $(\alpha E_2+(1-\alpha) E_4)$-critical for all $\alpha\in[0,1]$.

Note that if $\ph$ is inverse stereographic projection then $f'=r^{-1}\sin f$, so 
\beq\label{keanwa}
\ph^*\omega\cdot\ph^*\omega=-8f'(r)^4g.
\eeq
 We can use this simplification to
show that inverse stereographic projection is a {\em stable} restricted $(\alpha E_2+(1-\alpha) E_4)$-critical map for all $\alpha\in[0,1]$. For this, we need
the hessian associated with the Skyrme energy $E_4$.

\begin{prop} Let $\ph:(M,g)\ra N$ be restricted $E_4$-critical. Then the hessian of $E_4$ at $\ph$ is 
$$
\hess_{E_4}(X,Y)=-\frac12\ip{\L_X(\ph^*\omega\cdot\ph^*\omega),\L_Yg}_{L^2}+\frac12\ip{\ph^*\omega\cdot\L_Xg,\L_Yg\cdot\ph^*\omega}_{L^2}
$$
where $\ph^*\omega\cdot\L_Xg$ denotes the $\g$-valued bilinear form 
$$
(\ph^*\omega\cdot\L_Xg)(A,B)=\sum_i\ph^*\omega(A,e_i)(\L_Xg)(e_i,B),
$$
and similarly for $\L_Yg\cdot\ph^*\omega$.
\end{prop}

\begin{proof} Let $X,Y\in\Gamma_0(TM)$, $\psi_s$, $\chi_t$ be their 
flows, $\psi_{s,t}=\psi_s\circ\chi_t$ and $\ph_s=\ph\circ\psi_s$, as in the proof of Theorem \ref{secondvar}.
Then
\bea
\hess_{E_4}(X,Y)&=&\frac{\cd^2 E_4(\ph\circ\psi_{s,t},g)}{\cd s\: \cd t}\bigg|_{s=t=0}
=\frac{\cd^2 E_4(\ph_s\circ\chi_t,g)}{\cd s\: \cd t}\bigg|_{s=t=0}\nonumber\\
&=&\frac{\cd\:}{\cd s}\bigg|_{s=0}\frac{\cd\:}{\cd t}\bigg|_{t=0} E_4(\ph_s,\chi_{-t}^*g)\qquad
\mbox{(since $E_4$ is geometrically natural)}\nonumber\\
&=&\frac{d\: }{ds}\bigg|_{s=0}\frac12\ip{S_{E_4}(\ph_s,g),-\L_Yg}_{L^2}.
\eea
Now
\bea
S_{E_4}(\ph_s,g)&=&\frac12|\ph_s^*\omega|^2g+\ph_s^*\omega\cdot\ph_s^*\omega\nonumber \\
\Rightarrow\quad
\frac{d\: }{ds}\bigg|_{s=0}S_{E_4}(\ph_s,g)&=&fg+\frac{d\: }{ds}\bigg|_{s=0}[\psi_s^*(\ph^*\omega)\cdot\psi_s^*(\ph^*\omega)]\nonumber \\
&=&fg+(\L_X\ph^*\omega)\cdot\ph^*\omega+\ph^*\omega\cdot(\L_X\ph^*\omega),
\eea
and, by Remark \ref{netski}, $fg$ is pointwise orthogonal to $\L_Yg$ (since $\div Y=0$), so
\beq
\hess_{E_4}(X,Y)=-\frac12\ip{(\L_X\ph^*\omega)\cdot\ph^*\omega+\ph^*\omega\cdot(\L_X\ph^*\omega),\L_Yg}_{L^2}.
\eeq
Now
\bea
\L_X(\ph^*\omega\cdot\ph^*\omega)&=&\L_X\, c(\ph^*\omega\otimes\sharp g\otimes\ph^*\omega)\nonumber \\
&=&c(\L_X\ph^*\omega\otimes\sharp g\otimes\ph^*\omega+\ph^*\omega\otimes\L_X\sharp g\ph^*\omega+\ph^*\omega\otimes\sharp g\otimes\L_X\ph^*\omega)\nonumber\\
&=&(\L_X\ph^*\omega)\cdot\ph^*\omega-(\ph^*\omega\cdot\L_X g)\cdot\ph^*\omega+\ph^*\omega\cdot(\L_X\ph^*\omega),
\eea
where $c:T^2_4M\ra T^0_2M$ denotes the requisite contraction map. Hence
\bea
\hess_{E_4}(X,Y)&=&-\frac12\ip{\L_X(\ph^*\omega\cdot\ph^*\omega)+(\ph^*\omega\cdot\L_X g)\cdot\ph^*\omega,\L_Yg}_{L^2}\nonumber\\
&=&-\frac12\ip{\L_X(\ph^*\omega\cdot\ph^*\omega),\L_Yg}_{L^2}+\frac12\ip{\ph^*\omega\cdot\L_X g,\L_Yg\cdot\ph^*\omega}_{L^2},
\eea
where the sign change in the final term is caused by the antisymmetry of the bilinear form $\ph^*\omega$.
\end{proof}

\begin{prop}\label{kawrmc}
 Inverse stereographic projection $\ph:\R^3\ra S^3=SU(2)$ is restricted $E_4$-stable.
\end{prop}

\begin{proof} We must show that $\hess_{E_4}(X,X)\geq 0$ for all $X\in\Gamma_0(TM)$. Choose and fix $X\in\Gamma_0(TM)$.
At a given point $rn\in\R^3\less\{0\}$, where $r>0$ and
$n\in S^2$, choose and fix a unit vector $u\in T_nS^2$. Then, at $rn$ we can use the triple $e_1=\cd/\cd r$, $e_2=u/r$, $e_3=(n\times u)/r$ as
an orthonormal basis for $T_{rn}\R^3$, and $\eps_1=n$, $\eps_2=u$, $\eps_3=n\times u$ as an orthonormal basis for the Lie algebra
$\g=\su(s)=(\R^3,-2\times)$. Let $m_{ij}=(\L_Xg)(e_i,e_j)$, and note that, since $\div X=0$,  the matrix $m$ is traceless. Let 
$\xi_a(A,B)=\ip{\eps_a,\ph^*\omega(A,B)}_\g$ for all $A,B\in T_{rn}\R^3$, where $a=1,2,3$. Then, for example,
\beq
\xi_1(e_2,e_3)=n\cdot\ph^*\omega(u/r,n\times u/r)=-2r^{-2}\sin^2f=-2f'(r)^2
\eeq
by equation (\ref{kelannwat}). Computing all components $\xi_a(e_i,e_j)$ similarly, one finds that,
relative to the frame $\{e_1,e_2,e_3\}$ for $T_{rn}\R^3$,
\bea
\xi_1&=&2 f'(r)^2\left(\begin{array}{ccc}
0&0&0\\0&0&-1\\0&1&0
\end{array}\right),\nonumber\\
\xi_2&=&2 f'(r)^2\left(\begin{array}{ccc}
0&\sin f(r)&\cos f(r)\\-\sin f(r)&0&0\\-\cos f(r)&0&0
\end{array}\right),\nonumber\\
\xi_3&=&2 f'(r)^2\left(\begin{array}{ccc}
0&-\cos f(r)&\sin f(r)\\ \cos f(r)&0&0\\-\sin f(r)&0&0
\end{array}\right).\qquad
\eea
Hence, at the point $rn$,
\bea
\ip{\ph^*\omega\cdot\L_Xg,\L_Xg\cdot\ph^*\omega}&=&
\sum_a\sum_{i,j}(\xi_a\cdot\L_Xg)(e_i,e_j)(\L_Xg\cdot\xi_a)(e_i,e_j)\nonumber \\
&=&
\sum_{a,i,j,k,l}\xi_a(e_i,e_k)\L_Xg(e_k,e_j)\L_Xg(e_i,e_l)\xi_a(e_l,e_j)\nonumber\\
&=&-\sum_a\tr(\xi_am\xi_am)\nonumber\\
&=&-8f'(r)^4(m_{12}^2+m_{23}^2+m_{13}^2-m_{11}m_{22}-m_{22}m_{33}-m_{33}m_{11}).
\nonumber 
\eea
Comparing this (at the point $rn$) with
\beq
|\L_Xg|^2=2(m_{12}^2+m_{23}^2+m_{13}^2)+m_{11}^2+m_{22}^2+m_{33}^2
\eeq
and recalling that $m_{11}+m_{22}+m_{33}=0$, we see that, at the point $rn$,
\beq\label{kaw}
\ip{\ph^*\omega\cdot\L_Xg,\L_Xg\cdot\ph^*\omega}=-4f'(r)^4|\L_Xg|^2.
\eeq
Since the point $rn$ was arbitrary, (\ref{kaw}) holds on all $\R^3\less\{0\}$, and hence, by continuity on all $\R^3$. Recall (equation (\ref{keanwa}))
that $\ph^*\omega\cdot\ph^*\omega=-8f'(r)^4g$. Hence
\bea
\hess_{E_4}(X,X)&=&-\frac12\ip{\L_X(-8f'(r)^4g),\L_Xg}_{L^2}+\frac12\int_{\R^3}(-4f'(r)^4|\L_Xg|^2)\vol_g\nonumber\\
&=&2\|f'(r)^2\L_Xg\|_{L^2}^2
\eea
where we have, once again, used the fact that $\L_Xg$ is pointwise orthogonal to $g$ (Remark \ref{netski}).
\end{proof}

\begin{cor} Let $F_\alpha=\alpha E_2+(1-\alpha)E_4$. Then inverse stereographic projection $\R^3\ra S^3=SU(2)$ is a stable
restricted $F_\alpha$-critical map for all $\alpha\in[0,1]$.
\end{cor}

\begin{proof} Follows immediately from Corollary \ref{keanwarimyco} and Proposition \ref{kawrmc}.
\end{proof}

\section{Minimizing over $SL(k,\R)$}
\news
\label{sec-finite}

We have seen (Example \ref{barebell}) that the axially symmetric BPS skyrmions $\ph_B$ used in \cite{adasanwer,adasanwer2,adanaysanwer,bonmar,bonharmar} are not 
restricted harmonic (for $B\geq 2$), 
so do not minimize $E_2$ in their $\sdiff$ orbit. It is not clear how to construct the actual minimizer, or even
whether a minimizer exists. In this section we solve a finite-dimensional version of this variational problem, by minimizing $E_2$ over the
group of {\em linear} volume preserving diffeomorphisms of $\R^3$, $SL(3,\R)$.\footnote{Strictly speaking, $SL(3,\R)$ is not a subgroup of $\sdiff(\R^3)$, since
linear maps do not have compact support. However, given the boundary behaviour of $\ph_B$, its $SL(3,\R)$ orbit lies in the closure of its $\sdiff$ orbit for any 
reasonable choice of topology on the space of smooth maps $\R^3\ra SU(2)$.} We formulate the problem for a general map $\ph:\R^k\ra N$ with sufficiently good
boundary behaviour.

Given a fixed map $\ph:\R^k\ra N$, denote by $M$ its ``average strain matrix,'' the $k\times k$ matrix with entries
\beq
M_{ij}=\int_{\R^k}\ph^*h(\cd_i,\cd_j)d^kx.
\eeq

\begin{prop}\label{opti}
Let $\ph:\R^k\ra N$ be a nonconstant map with $\ph(x)\ra\ph_\infty$, constant, as $|x|\ra\infty$ sufficiently fast that its
average strain matrix $M$ is finite. Let $E_\ph:SL(k,\R)\ra\R$ be the function $E_\ph(A)=E_2(\ph\circ A)$. 
Then
$$
E_\ph(A)\geq\frac{k}{2}(\det M)^{1/k}
$$
with equality if and only if $A^TMA=\mu\I_k$ for some $\mu>0$, and this equality is always attained.
\end{prop}

\begin{proof}
Note that $M$ is symmetric and non-negative. If $Mv=0$ with $v\neq 0$ then $\d\ph(v)=0$ everywhere, so $\ph$ is constant in the direction of $v$, hence
constant (by the boundary condition), which is false by assumption. Hence $M$ is positive definite. Now
 the Dirichlet energy of $\ph\circ A$ is
\beq
E_2(\ph\circ A)=
\frac12\int_{\R^k}\sum_iA^*\ph^*h(\cd_i,\cd_i)d^kx=
\frac12\int_{\R^k}\sum_i \ph^*h(\sum_jA_{ji}\cd_i,\sum_jA_{ji}\cd_i)d^kx=\frac12\tr(A^TMA).
\eeq
Hence, by the AM-GM inequality applied to the eigenvalues of $A^TMA$, 
\beq
E_2(\ph\circ A)\geq \frac{k}{2}(\det A^TMA)^{1/k}=\frac{k}{2}(\det M)^{1/k}
\eeq
with equality if and only if $A^TMA=\mu\I_k$ for some $\mu\geq 0$. It remains to show that equality always occurs.
 Let $\{e_i,\: :\: i=1,\ldots,k\}$ be an oriented orthonormal basis of eigenvectors of $M$, with eigenvalues $\lambda_i^2>0$, $O_M$ be
the $SO(k)$ matrix with columns $(e_1,\ldots,e_k)$ and $D_M=\diag(\lambda_1^{-1},\ldots,\lambda_k^{-1})$. Then
\beq\label{A0}
A_0=(\det M)^{1/2k}O_MD_M
\eeq
is in $SL(k,\R)$ and 
$A_0^TMA_0=(\det M)^{1/k}\I_k$.
\end{proof}

\begin{remark} The condition for $\ph$ to minimize $E_2$ over its $SL(k,\R)$ orbit ($M=\mu\I_k$) is precisely the condition that the
trace-free part of its stress tensor be $L^2$ orthogonal to all variations of the metric $g$ through constant coefficient
metrics \cite{spe-crystals}. Alternatively, it is the condition that $\ph$ should satisfy all the extended Derrick
identities \cite{man-der} for the energy $E_0+\eps E_2+E_6$ except the basic one, generated by dilations of $\R^k$.
\end{remark}

\begin{eg}[Linearly improved BPS skyrmions] Consider the BPS skyrmion $\ph_B:\R^3\ra S^3$ constructed above in Example \ref{barebell} (equation
(\ref{sayosuco})). This has
\bea
\ph_B^*h&=&\left(f_B'(r)^2-\frac{\sin^2f_B(r)}{r^2}\right)\frac{(\xvec\cdot d\xvec)^2}{r^2}+\frac{\sin^2 f_B(r)}{r^2}d\xvec\cdot d\xvec+\nonumber \\
&&\qquad(B^2-1)\sin^2f_B(r)\left(\frac{x_1dx_2-x_2dx_1}{x_1^2+x_2^2}\right)^2
\eea
so its average strain matrix is
\bea
M&=&2\pi B^{1/3}\left\{\left(\frac23C_1+\frac43 C_2\right)\I_3+(B^2-1)C_2\left(\begin{array}{ccc}1&0&0\\0&1&0\\0&0&0\end{array}\right)\right\}
\eea
where
\beq\label{Cdef}
C_1=\int_0^\infty f'(r)^2r^2\, dr,\quad C_2=\int_0^\infty\sin^2f(r)\, dr
\eeq
and $f_B(r)=f(B^{-1/3}r)$. It follows from Proposition \ref{opti} that the minimum of $E_2$ over the
$SL(3,\R)$ orbit of $\ph_B$ is
\beq
E_2(\ph_B\circ A_B)=\frac32(\det M)^{1/3}=2\pi B^{1/3}\left(C_1+\frac{1+3B^2}{2}C_2\right)^{2/3}\left(C_1+2C_2\right)^{1/3}
\eeq
attained at
\beq\label{AB}
A_B=\left(\begin{array}{ccc}\lambda_B&0&0\\0&\lambda_B&0\\0&0&\lambda_B^{-2}\end{array}\right),\qquad
\lambda_B=\left(\frac{2C_1+4C_2}{2C_1+(3B^2+1)C_2}\right)^{1/6}.
\eeq
Note that $E_2(\ph_B\circ A_B)$ grows like $B^{5/3}$, whereas $E_2(\ph_B)=\frac12\tr M$ grows like $B^{7/3}$, so 
the energy saved by deforming $\ph_B$ along $SL(3,\R)$ grows
without bound as $B$ increases. Note also that $\lambda_1=1$ and $\lambda_B$ decreases monotonically towards $0$, so $A_B$ has the effect of spreading $\ph_B$ in the $x_1x_2$ plane while squashing it in the $x_3$ direction, and this distortion grows more extreme with larger $B$. This suggests that the approximation of
near BPS skyrmions by $\ph_B$ gets progressively worse as $B$ increases. Since $E_2(\ph_B\circ A_B)$ still grows faster than $B$, it is clear that linearly
improved BPS skyrmions are not competitive candidates to approximate near BPS skyrmions at large $B$, as they are more energetic than $B$ remote superposed charge 1
BPS skyrmions. For appropriate choices of potential $U$, they are energetically favoured over charge 1 clusters for low $B$, however.
\end{eg}

\subsection*{Appendix: proof of Lemma \ref{gsos}}
\label{appendixA}
\news
\renewcommand{\theequation}{A.\arabic{equation}}

We compute, using the definitions,
\bea
\div\iota_X\alpha&=&\sum_i\left(e_i[\alpha(X,e_i)]-\alpha(X,\nabla_{e_i}e_i)\right),\nonumber \\
(\div\alpha)(X)&=&\sum_i\left(e_i[\alpha(X,e_i)]-\alpha(\nabla_{e_i}e_i,X)-\alpha(e_i,\nabla_{e_i}X)\right)\nonumber \\
\Rightarrow\quad\div\iota_X\alpha-(\div \alpha)(X)&=&\sum_i\alpha(e_i,\nabla_{e_i}X),
\eea
and
\bea
\ip{\alpha,\L_Xg}_g&=&\sum_{i,j}\alpha(e_i,e_j)(\L_Xg)(e_i,e_j)\nonumber \\
&=&\sum_{i,j}\alpha(e_i,e_j)\left(X[g(e_i,e_j)]-g(\L_Xe_i,e_j)-g(e_i,\L_Xe_j)\right)\nonumber \\
&=&-\sum_{i,j}\alpha(e_i,e_j)\left(g(\nabla_Xe_i-\nabla_{e_i}X,e_j)+g(e_i,\nabla_Xe_j-\nabla_{e_j}X)\right)\nonumber \\
&=&-\sum_{i,j}\alpha(e_i,e_j)\left(X[g(e_i,e_j)]-g(\nabla_{e_i}X,e_j)-g(e_i,\nabla_{e_j}X)\right)\nonumber \\
&=&2\sum_{i,j}\alpha(e_i,e_j)g(e_j,\nabla_{e_i}X)\nonumber \\
&=&2\sum_i\alpha(e_i,\sum_jg(e_j,\nabla_{e_i}X)e_j)\nonumber \\
&=&2\sum_i\alpha(e_i,\nabla_{e_i}X)
\eea
which completes the proof.

\subsection*{Acknowledgements}

The author wishes to thank Christoph Adam, Andrzej Wereszczy\'nski and Nick Manton for useful discussions. This work was supported by the UK
Engineering and Physical Sciences Research Council.

\end{document}